\documentclass[12pt]{amsart}

\usepackage{amscd,amsmath}
\usepackage{amssymb,amsmath,latexsym}

\usepackage[all]{xy}
\usepackage[pdftex]{graphicx}

\newtheorem{theorem}{Theorem}[section]
\newtheorem{cor}[theorem]{Corollary}
\newtheorem{lemma}[theorem]{Lemma}

\begin{document}

\title{Distance-Two Colorings of Barnette Graphs}

\author[T. Feder]{Tom\'{a}s Feder}
\address{268 Waverley Street \\
               Palo Alto, CA 94301, USA}
\email{tomas@theory.stanford.edu}

\author[P. Hell]{Pavol Hell}
\address{School of Computing Science \\
 	      Simon Fraser University \\
              Burnaby, B.C., Canada V5A 1S6}
\email{pavol@sfu.ca}

\author[C.Subi]{Carlos Subi}
\address{Los Altos Hills}
\email{carlos.subi@hotmail.com}

\date{}

\maketitle

\begin{abstract}
Barnette identified two interesting classes of cubic polyhedral graphs for which he conjectured 
the existence of a Hamiltonian cycle. Goodey proved the conjecture for the intersection of the
two classes. We examine these classes from the point of view of distance-two colorings. A 
distance-two $r$-coloring of a graph $G$ is an assignment of $r$ colors to the vertices of 
$G$ so that any two vertices at distance at most two have different colors. Note that a cubic
graph needs at least four colors. The distance-two four-coloring problem for cubic planar 
graphs is known to be NP-complete. We claim the problem remains NP-complete for 
tri-connected bipartite cubic planar graphs, which we call type-one Barnette graphs, 
since they are the first class identified by Barnette. By contrast, we claim the problem
is polynomial for cubic plane graphs with face sizes $3, 4, 5,$ or $6$, which we call 
type-two Barnette graphs, because of their relation to Barnette's second conjecture.
We call Goodey graphs those type-two Barnette graphs all of whose faces have size $4$ 
or $6$. We fully describe all Goodey graphs that admit a distance-two four-coloring, 
and characterize the remaining type-two Barnette graphs that admit a distance-two 
four-coloring according to their face size.

For quartic plane graphs, the analogue of type-two Barnette graphs are graphs with 
face sizes $3$ or $4$. For this class, the distance-two four-coloring problem is also 
polynomial; in fact, we can again fully describe all colorable instances -- there are 
exactly two such graphs.
\end{abstract}

\section{Introduction}

Tait conjectured in 1884 \cite{tait} that all cubic polyhedral graphs, i.e., all tri-connected cubic planar graphs, 
have a Hamiltonian cycle; this was disproved by Tutte in 1946 \cite{tutte}, and the study of Hamiltonian cubic 
planar graphs has been a very active area of research ever since, see for instance \cite{aldred,npc,derek,lu}. 
Barnette formulated two conjectures that have been at the centre of much of the effort: (1) that {\em bipartite} 
tri-connected cubic planar graphs are Hamiltonian (the case of Tait's conjecture where all face sizes are even) 
\cite{B}, and (2) that tri-connected cubic planar graphs with all face sizes $3, 4, 5$ or $6$ are Hamiltonian, cf.
\cite{b2,m}. Goodey \cite{good,goodtwo} proved that the conjectures hold on the intersection of the two 
classes, i.e., that tri-connected cubic planar graphs with all face sizes $4$ or $6$ are Hamiltonian. When 
all faces have sizes $5$ or $6$, this was a longstanding open problem, especially since these graphs 
(tri-connected cubic planar graphs with all face sizes $5$ or $6$) are the popular {\em fullerene graphs} \cite{rok}.
The second conjecture has now been affirmatively resolved in full \cite{kardos}. For the first conjecture, 
two of the present authors have shown in \cite{tomas} that if the conjecture is false, then the Hamiltonicity 
problem for tri-connected cubic planar graphs is NP-complete. In view of these results and conjectures, in 
this paper we call bipartite tri-connected cubic planar graphs {\em type-one Barnette graphs}; we call 
cubic plane graphs with all face sizes $3, 4, 5$ or $6$ {\em type-two Barnette graphs}; and finally we 
call cubic plane graphs with all face sizes $4$ or $6$ {\em Goodey graphs}. Note that it would be more 
logical, and historically accurate, to assume tri-connectivity also for type-two Barnette graphs and for 
Goodey graphs. However, we prove our positive results without needing tri-connectivity, and hence we 
do not assume it.

Cubic planar graphs have been also of interest from the point of view of colorings \cite{kostochkasurvey,havet}. 
In particular, they are interesting for distance-two colourings. Let $G$ be a graph with degrees at most $d$.
A {\em distance-two $r$-coloring} of $G$ is an assignment of colors from $[r]=\{1,2,\ldots,r\}$ to the vertices
of $G$ such that if a vertex $v$ has degree $d'\leq d$ then the $d'+1$ colors of $v$ and of all the neighbors 
of $v$ are all distinct. (Thus a distance-two coloring of $G$ is a classical coloring of $G^2$.) Clearly a graph
with maximum degree $d$ needs at least $d+1$ colors in any distance-two coloring, since a vertex of 
degree $d$ and its $d$ neighbours must all receive distinct colors. It was conjectured by Wegner \cite{weg} 
that a planar graph with maximum degree $d$ has a distance-two $r$-colouring where $r=7$ for $d=3$, 
$r=d+5$ for $d=4, 5, 6, 7$, and $r=\lfloor 3d/2 \rfloor + 1$ for all larger $d$. The case $d=3$ has been 
settled in the positive by Hartke, Jahanbekam and Thomas \cite{hjt}, cf. also \cite{carsten}.

For cubic planar graphs in general it was conjectured in \cite{hjt} that if a cubic planar graph 
is tri-connected, or has no faces of size five, then it has a distance-two six-coloring. We propose 
a weaker version of the second case of the conjecture, namely, we conjecture that {\em a bipartite 
cubic planar graph can be distance-two six-colored}. We prove this in one special case (Theorem
\ref{coze}), which of course also confirms the conjecture of Hartke, Jahanbekam and Thomas 
for that case. Heggerness and Telle \cite{ht} have shown that the problem of distance-two four-coloring 
cubic planar graphs is NP-complete. On the other hand, Borodin and Ivanova \cite{borodin} have shown 
that subcubic planar graphs of girth at least $22$ can be distance-two four-colored. In fact, there has
been much attention focused on the relation of distance-two colorings and the girth, especially in the
planar context \cite{borodin,havet}.

Our results focus on distance-two colorings of cubic planar graphs, with particular attention on Barnette 
graphs, of both types. We prove that a cubic plane graph with all face sizes divisible by four can always
be distance-two four-colored, and a give a simple condition for when a bi-connected cubic plane graph 
with all face sizes divisible by three can be distance-two four-colored using only three colors per face.
It turns out that the distance-two four-coloring problem for type-one Barnette graphs is NP-complete,
while for type-two Barnette graphs it is not only polynomial, but the positive instances can be explicitly
described. They include one infinite family of Goodey graphs (cubic plane graphs with all faces of size
$4$ or $6$), and all type-two Barnette graphs which have all faces of size $3$ or $6$. Interestingly,
there is an analogous result for quartic (four-regular) graphs: all quartic planar graphs with faces of
only sizes $3$ or $4$ that have a distance-two five coloring can be explicitly described; there are only
two such graphs.

Note that we use the term ``plane'' graph when the actual embedding is used, e.g., by discussing 
the faces; when the embedding is unique, as in tri-connected graphs, we stick with writing ``planar".

\section{Relations to edge-colorings and face-colorings}

Distance-two colorings have a natural connection to edge-colorings.

\begin{theorem}
Let $G$ be a graph with degrees at most $d$ that admits a distance-two $(d+1)$-coloring, 
with $d$ odd. Then $G$ admit an edge-coloring with $d$ colors.
\end{theorem}

\begin{proof}
The even complete graph $K_{d+1}$ can be edge-colored with $d$ colors by the Walecki 
construction~\cite{wal}. We fix one such coloring $c$, and then consider a distance-two $(d+1)$-coloring
of $G$. If an edge $uv$ of $G$ has colors $ab$ at its endpoints, we color $uv$ in $G$ with the color $c(ab)$. 
It is easy to see that this yields an edge-coloring of $G$ with $d$ colors.
\end{proof}

We call the resulting edge-coloring of $G$ the {\em derived edge-coloring} of the original distance-two coloring.

In this paper, we mostly focus on the case $d=3$ (the {\em subcubic} case). Thus we use the edge-coloring 
of $K_4$ by colors red, blue, green. This corresponds to the unique partition of $K_4$ into perfect matchings. 
Note that for every vertex $v$ of $K_4$ and every edge-color $i$, there is a unique other vertex $u$ of $K_4$ 
adjacent to $v$ in edge-color $i$. Thus if we have the derived edge-coloring we can efficiently recover the 
original distance-two coloring. In the subcubic case, in turns out to be sufficient to have just one color
class of the edge-coloring of $G$.

\begin{theorem}\label{twotwo}
Let $G$ be a subcubic graph, and let $R$ be a set of red edges in $G$.
The question of whether there exists a distance-two four-coloring of $G$
for which the derived edge-coloring has $R$ as one of the three color classes
can be solved by a polynomial time algorithm. If the answer is positive, the
algorithm will identify such a distance-two coloring.
\end{theorem}

\begin{proof}
We may assume in $K_4$ red joins colors $13, 24,$ blue joins colors $12, 34$ and
green joins colors $14, 23$. Note that we may also assume that $R$ is a matching 
that covers at least all vertices of degree three, otherwise we answer in the negative.
We may further assume that some vertex $v$ gets an even color ($2$ or $4$). 
The parity of the color of a vertex $u$ determines the 
parity of the color of its neighbors, namely the parity is the same if they are adjacent 
by an edge in $R$, and they are of different parity otherwise. We may thus extend
from $v$ the assignment of parities to all the vertices, unless an inconsistency is 
reached, in which case no coloring exists. Otherwise, at this point all vertices have 
only two possible colors, namely $1, 3$ for odd and $2, 4$ for even.

Define an auxiliary graph $G'$ with vertices $V(G')=V(G)$, and edges $xy$ in $E(G')$ 
if $xy$ is a red edge in $E(G)$ or if there is a path $xzy$ without red edges in $E(G)$. 
Note that these edges $xy$ join vertices of the same parity, and $x, y$ must have different
colors. If $G'$ has an odd cycle, then no solution exists. Otherwise $G'$ is bipartite, and 
we may choose $1, 3$ in different sides of a bipartition of $G'$ for odd vertices, and $2, 4$ 
in different sides for even vertices.

Each vertex $u$ will have at most one neighbor $x$ of the same parity in $G$,
namely the one joined to it by the red edge, and $ux$ is an edge of $G'$. This 
guarantees different colors for $u, x$. The at most two other neighbors $y, z$ 
of $u$ have different parity from $u, x$, and the path $yuz$ in $G$ ensures the
edge $yz$ is in $G'$. This guarantees different colors for $y, z$. Thus the colors
for $u, x, y, z$ are all different at each vertex $u$, and we have a distance-two coloring
of $G$.
\end{proof}

There is also a relation to face-colorings. It is a folklore fact that the faces of any bipartite cubic plane graph $G$
can be three-colored \cite{ore}. This three-face-coloring induces a three-edge-coloring of $G$ by coloring each 
edge by the color not used on its two incident faces.
(It is easy to see that this is in fact an edge-coloring, i.e., that incident
edges have distict colors.) We call an edge-coloring that arises this way from some
face-coloring of $G$ a {\em special three-edge-coloring} of $G$. We first ask 
when is a special three-edge-coloring of $G$ the derived edge-coloring of a 
distance-two four-coloring of $G$.

\begin{theorem}\label{threethree}
A special three-edge-coloring of $G$ is the derived edge-coloring of some 
distance-two four-coloring of $G$ if and only if the size of each face is a 
multiple of $4$.
\end{theorem}

\begin{proof}
The edges around a face $f$ alternate in colors, and the vertices of $f$ can
be colored consistently with this alternation if and only if the size of $f$
is a multiple of 4. This proves the ``only if'' part. For the ``if'' part, suppose
all faces have size multiple of $4$. If there is an inconsistency, it will
appear along a cycle $C$ in $G$. If there is only one face inside $C$,
there is no inconsistency. Otherwise we can join some two vertices of $C$ by
a path $P$ inside $C$, and the two sides of $P$ inside $C$ give two regions
that are inside two cycles $C', C''$. The consistency of $C$ then follows from
the consistency of each of $C', C''$ by induction on the number of faces inside
the cycle.
\end{proof}

\begin{cor}
Let $G$ be a cubic plane graph in which the size of each face is a multiple of four.
Then $G$ can be distance-two four-colored.
\end{cor}

We now prove a special case of the conjecture stated in the introduction, that
all bipartite cubic plane graphs can be distance-two six-colored. Recall that the
faces of any bipartite cubic plane graph can be three-colored.

\begin{theorem}\label{coze}
Suppose the faces of a bipartite cubic plane graph $G$ are three-colored red, 
blue and green, so that the red faces are of arbitrary even size, while the size
of each blue and green face is a multiple of $4$. Then $G$ can be distance-two 
six-colored.
\end{theorem}

\begin{proof}
Let $G'$ be the multigraph obtained from $G$ by shrinking each of the red faces. 
Clearly $G'$ is planar, and since the sizes of blue and green faces in $G'$
are half of what they were in $G$, they will be even, so $G'$ is also bipartite.
Let us label the two sides of the bipartition as $A$ and $B$ respectively. 
Now consider the special three-edge coloring of $G$ associated with the
face coloring of $G$. Each red edge in this special edge-coloring joins a
vertex of $A$ with a vertex of $B$; we orient all red edges from $A$ to $B$. 
Now traversing each red edge in $G$ in the indicated orientation either has 
a blue face on the left and green face on the right, or a green face on the left 
and blue face on the right. In the former case we call the edge {\em class one}
in the latter case we call it {\em class two}. Each vertex of $G$ is incident
with exactly one red edge; the vertex inherits the class of its red edge. The 
vertices around each red face in $G$ are alternatingly in class 1 and class 2. 
We assign colors $1, 2, 3$ to vertices of class one and colors $4, 5, 6$ to vertices 
of class two. It remains to decide how to choose from the three colors available 
for each vertex. A vertex adjacent to red edges in class $i$ has only three vertices 
within distance two in the same class, namely the vertex across the red edge, 
and the two vertices at distance two along the red face in either direction.
Therefore distance-two coloring for class $i$ corresponds to three-coloring 
a cubic graph. Since neither class can yield a $K_4$, such a three-coloring
exists by Brooks' theorem~\cite{br}. This yields a distance-two six-coloring 
of $G$.
\end{proof}

\section{Distance-two four-coloring of type-one Barnette graphs is NP-complete}

We now state our main intractability result. 

\begin{theorem}\label{tuza}
The distance-two four-coloring problem for tri-connected bipartite cubic planar graphs 
is NP-complete.
\end{theorem}

We will begin by deriving a weaker version of our claim.

\begin{theorem}\label{here}
The distance-two four-coloring problem for bipartite planar subcubic graphs
is NP-complete.
\end{theorem}

\begin{proof}
Consider the graph $H$ in Figure \ref{ftent}.

\begin{figure}[h!]
\begin{center}
\includegraphics[height=4cm]{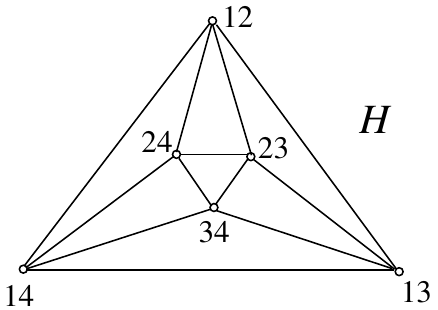}
\end{center}
\caption{The graph $H$ for the proof of Theorem \ref{here}}
\label{ftent}
\end{figure}

We will reduce the problem of $H$-coloring planar graphs to the distance-two 
four-coloring problem for bipartite planar subcubic graphs. In the {\em $H$-coloring
problem} we are given a planar graph $G$ and and the question is whether we
can color the vertices of $G$ with colors that are vertices of $H$ so that adjacent 
vertices of $G$ obtain adjacent colors. This can be done if and only if $G$ is
three-colorable, since the graph $H$ both contains a triangle and is three-colorable
itself. (Thus any three-coloring of $G$ is an $H$-coloring of $G$, and any $H$-coloring
of $G$ composed with a three-coloring of $H$ is a three-coloring of $G$.) It is known
that the three-coloring problem for planar graphs is NP-complete, hence so is the
$H$-coloring problem.

\begin{figure}[h!]
\begin{center}
\includegraphics[height=5cm]{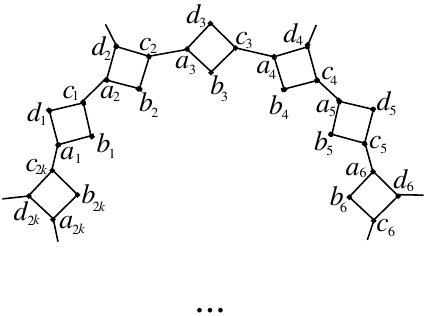}
\end{center}
\caption{The ring gadget}
\label{ring}
\end{figure}

Thus suppose $G$ is an instance of the $H$-coloring problem. We form a new graph
$G'$ obtained from $G$ by replacing each vertex $v$ of $G$ by a {\em ring} gadget depicted
in Figure \ref{ring}. If $v$ has degree $k$, the ring gadget has $2k$ squares. A  {\em link} in 
the ring is a square $a_ib_ic_id_ia_i$ followed by the edge $c_ia_{i+1}$. A link is {\em even}
if $i$ is even, and {\em odd} otherwise. Every even link in the ring will be used for a connection
to the rest of the graph $G'$, thus vertex $v$ has $k$ available links. For each edge $vw$
of $G$ we add a new vertex $f_{vw}$ that is adjacent to a vertex $d_s$ in one available
link of the ring for $v$ and a vertex $d'_t$ in one available link of the ring for $w$. (We use
primed letters for the corresponding vertices in the ring of $w$ to distinguish them from those
in the ring of $v$.) The actual choice of (the even) subscripts $s, t$ does not matter, as long 
as each available link is only used once. The resulting graph is clearly subcubic and planar.
It is also bipartite, since we can bipartition all its vertices into one independent set $A$ 
consisting of all the vertices $a_i, c_i, b_{i+1}, d{i+1}$ with odd $i$ in all the rings, and 
another independent set $B$ consisting of the vertices $a_i, c_i, b_{i+1}, d{i+1}$ with 
even $i$ in all the rings. Moreover, we place all vertices $f_{vw}$ into the set $A$.
Note that in any distance-two four-coloring of the ring, each link must have four different 
colors for vertices $a_i, b_i, c_i, d_i$, and the same color for $a_i$ and $a_{i+1}$. Thus all
$a_i$ have the same color and all $c_i$ have the same color. The pair of colors in $b_i, d_i$
is also the same for all $i$; we will call it the {\em characteristic pair of the ring for $v$}.
For any pair $ij$ of colors from $1, 2, 3, 4$, there is a distance-two coloring of the ring that
has the characteristic pair $ij$.

We prove that $G$ is $H$-colorable if and only if $G'$ is distance-two four-colorable.
In an $H$-coloring of $G$, the vertices of $G$ are actually assigned unordered pairs 
from $\{1, 2, 3, 4\}$, since the vertices of $H$ are labeled by pairs. (Note that two 
vertices of $H$ are adjacent if and only if the pairs they are labeled with intersect 
in exactly one element.) Thus suppose that we have an $H$-coloring $\phi$ of $G$.
If $\phi(v)=ij$ (i.e., the vertex $v$ of $G$ is assigned the vertex of $H$ labeled by the pair 
$ij$), then we colour the ring of $v$ so that its characteristic pair is $ij$. This still leaves
a choice of which of the colors $i, j$ is in which $b_s, d_s$, in each of the links 
$a_s, b_s, c_s, d_s$. Since $\phi$ is an $H$ coloring, adjacent vertices $vw$ are
assigned pairs that intersect if exactly one element. This makes it possible to color
each $b_s, d_s$ so that all colors at distance at most two are distinct. For instance
if vertices $v$ and $w$ are adjacent in $G$ and colored by $12, 13$ by $\phi$, and
if $f_{vw}$ is adjacent to the vertices $d_s$ in the ring for $v$ and $d_t$ in the ring 
for $w$, then both $b_s$ in the ring for $v$ and $b_t$ in the ring for $w$ are colored 
$1$, as is $f_{vw}$, while $d_s$ in the ring for $v$ and $d_t$ in the ring for $w$ are 
colored $2$ and $3$ respectively. It is easy to see that this is a distance-two four-coloring
of $G'$.

Conversely, in any distance-two four-coloring of $G'$, the color of a vertex $f_{vw}$
determines the same color in the $b$'s of its adjacent links of the rings for $v$ and $w$,
whence the characteristic pairs of these two rings intersect in exactly one element. Thus
we may define a mapping $\phi$ of $V(G)$ to $V(H)$ by assigning to each vertex $v \in V(G)$
the characteristic pair of the ring for $v$. Then $\phi$ is an $H$-coloring of $G$, since adjacent 
vertices of $G$ are assigned pairs that are adjacent in $H$.
\end{proof}

\begin{figure}[htb!]
\includegraphics[width=6cm]{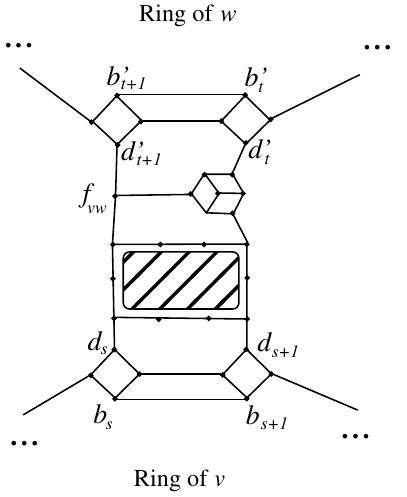}
\caption{The modified edge-gadget for $f_{vw}$}
\label{dump3}
\end{figure}

\vspace{3mm}

\begin{figure}[hbt!]
\includegraphics[width=6cm]{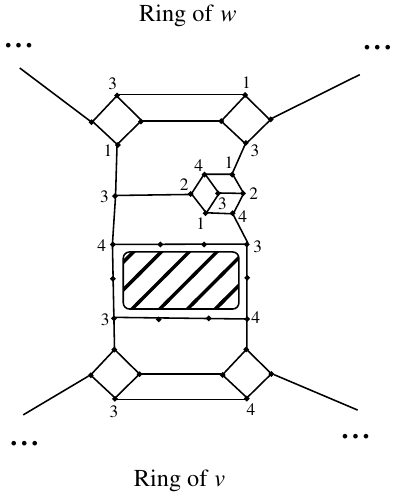}
\caption{With its unique distance-two coloring}
\label{dump2}
\end{figure}

To prove the full Theorem \ref{tuza}, the construction of the graph $G'$ is modified as suggested in Figure \ref{dump3}. 
Recall that in the construction of $G'$, for each edge $vw$ of $G$ a separate vertex $f_{vw}$ was 
made adjacent to $d_s$ in the ring of $v$ and $d'_t$ in the ring of $w$. Recall that both $s$ and $t$ 
are even, and the vertices $d_{s+1}, d'_{t+1}$ (with both subscripts odd) remained available for 
connection. We now make a new edge-gadget around the vertex $f_{vw}$, making it directly adjacent 
to $d'_{t+1}$, and connected to $d_s$ by a path, as depicted in Figure \ref{dump3}. In both rings, the 
two ``$b$" type vertices in the two consecutive links are joined together by an additional edge; 
specifically, we add the edges $b_sb_{s+1}$ and $b'_tb'_{t+1}$. 
(Note that this forces the corresponding ``$d$" type vertices $d_s$ and $d_{s+1}$ to be colored differently 
in any distance-two four-coloring, and similarly for $d'_t$ and $d'_{t+1}$). Moreover, further vertices and 
edges are added, as depicted in Figure \ref{dump3}. The shaded ten-sided region is identified with the 
ten-sided exterior face of the graph depicted in Figure \ref{fGoodey}, which has a unique distance-two 
four-coloring, shown there. (The heavy edges correspond to the ten-sided shaded figure.)
(Note that the graph in Figure \ref{fGoodey} was obtained from the graph in Figure \ref{Goodeyc1} 
by the deletion of two edges.)
Note that the construction is not symmetric, as it depends on which ring is viewed as the 
``bottom" ring for the vertex $f_{vw}$. (The depicted figure has the ring of $v$ on the bottom, 
but the conclusions are the same if it were the ring of $w$.) 
We can choose either way, independently for each edge $vw$ of $G$.
It can be seen that the resulting graph, which we denote by $G''$, is bipartite, planar, and cubic. 
We may assume that $G$ is bi-connected (the 
three-coloring problem for biconnected planar graphs is still NP-complete), and 
therefore $G''$ is also tri-connected (as no two faces share more than one edge). Using 
the unique distance-two four-colouring of the graph in Figure \ref{fGoodey}, it also follows that 
in any distance-two four-coloring of $G''$ the vertices $d_s$ and $d'_{t+1}$ have different
colors, while both vertices $b_s$ and $b'_{t+1}$ have the same color (the color of $f_{vw}$), 
in any distance-two four-coloring of $G''$. To facilitate checking this, we show in Figure \ref{dump2} 
a partial distance-two four-coloring, by circles, squares, up triangles, and 
down triangles; this coloring is forced by arbitrarily coloring $f_{vw}$ and its three neighbours
by four distinct colors. Since the colors of the pair $b_s, d_s$ and the pair $b'_{t+1}, d'_{t+1}$ 
have exactly one color in common, the previous NP-completeness proof applies, i.e., $G$
is $H$-colorable if and only if $G''$ is distance-two four-colorable.

\begin{figure}[h!]
\begin{center}
\includegraphics[height=4cm]{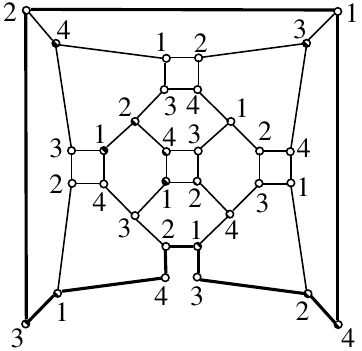}
\end{center}
\caption{The graph for the shaded region, with its unique distance-two four-coloring}
\label{fGoodey}
\end{figure}

We remark that (with some additional effort) we can prove that the problem is still NP-complete
for the class of tri-connected bipartite cubic planar graphs with no faces of sizes larger than $44$.

\section{Distance-two four-coloring of Goodey graphs}\label{dobra}

Recall that Goodey graphs are type-two Barnette graph with all faces of size $4$ and $6$ \cite{good,goodtwo}. 
In other words, a {\em Goodey graph} is a cubic plane graph with all faces having size $4$ or $6$. By Euler's 
formula, a Goodey graph has exactly six square faces, while the number of hexagonal faces is arbitrary. 

A {\em cyclic prism} is the graph consisting of two disjoint even cycles $a_1a_2\cdots a_{2k}a_1$ and
$b_1b_2\cdots b_{2k}b_1, k \geq 2,$ with the additional edges $a_ib_i$, $1\leq i\leq 2k$. It is easy to 
see that cyclic prisms have either no distance-two four-coloring (if $k$ is odd), or a unique distance-two 
four-coloring (if $k \geq 2$ is even). Only the cyclic prisms with $k=2, 3$ are Goodey graphs, and thus
from Goodey cyclic graphs only the cube (the case of $k=2$) has a distance-two coloring, which is 
moreover unique.

In fact, all Goodey graphs that admit distance-two four-coloring can be constructed from the cube as
follows. The Goodey graph $C_0$ is the cube, i.e., the cyclic prism with $k=2$. The Goodey graph $C_1$
is depicted in Figure \ref{Goodeyc1}. It is obtained from the cube by separating the six square faces and joining
them together by a pattern of hexagons, with three hexagons meeting at a vertex tying together the three
faces that used to meet in one vertex. The higher numbered Goodey graphs are obtained by making the 
connecting pattern of hexagons higher and higher. The next Goodey graph $C_2$ has two hexagons
between any two of the six squares, with a central hexagon in the centre of any three of the squares, 
the following Goodey graph $C_3$ has three hexagons between any two of the squares and three 
hexagons in the middle of any three of the squares, and so on. Thus in general we replace every vertex 
of the cube by a triangular pattern of hexagons whose borders are replacing the edges of the cube. We 
illustrate the vertex replacement graphs in Figure \ref{messy}, without giving a formal description. The 
entire Goodey graph $C_1$  is depicted in Figure \ref{Goodeyc1}.

\begin{figure}[h!]
\includegraphics[height=5.5cm]{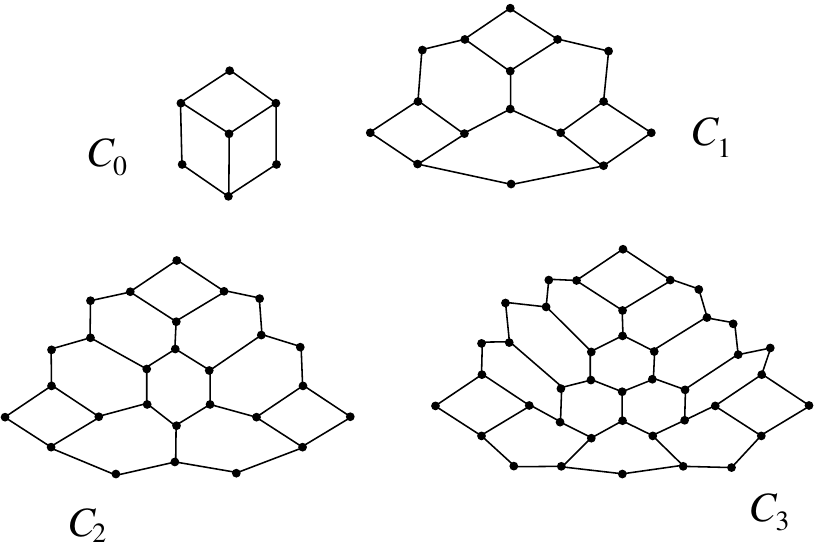}
\caption{The vertex replacements for Goodey graphs $C_0, C_1, C_2,$ and $C_3$}
\label{messy}
\end{figure}

\begin{figure}[h!]
\begin{center}
\includegraphics[height=4cm]{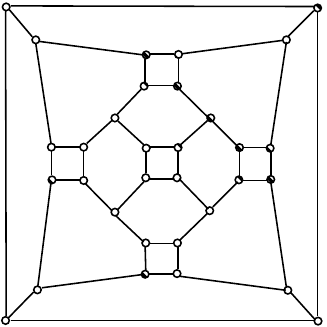}
\end{center}
\caption{The Goody graph $C_1$}
\label{Goodeyc1}
\end{figure}

We have the following results.

\begin{theorem}\label{cd}
The Goodey graphs $C_k, k\geq 0,$ have a unique distance-two four-coloring,
up to permutation of colors.
\end{theorem}

\begin{proof}
We described $C_k$ as eight triangular regions $R$, each consisting of
$k \choose 2$ hexagons, one region $R$ for each vertex of the cube. Each 
$R$ has three squares at the corners, which we describe as two squares 
joined by a chain of $k$ hexagons horizontally at the bottom, and a third 
square on top. (See Figure \ref{messy}.)

We partition the vertices into $k+2$ horizontal paths $P_i$, $0\leq i\leq k+1$,
with each $P_i$ having endpoints of degree 2 and internal vertices of degree 3.
The path $P_0$ has length $2k+2$, and the remaining paths $P_i$, $i\geq 1$
have length $2k+6-2i$. In particular the last $P_{k+1}$ has length 4, and is
the only $P_i$ that is actually a cycle, pictured as the square at the top. See
Figure \ref{messay2}.

\begin{figure}[h!]
\includegraphics[height=7cm]{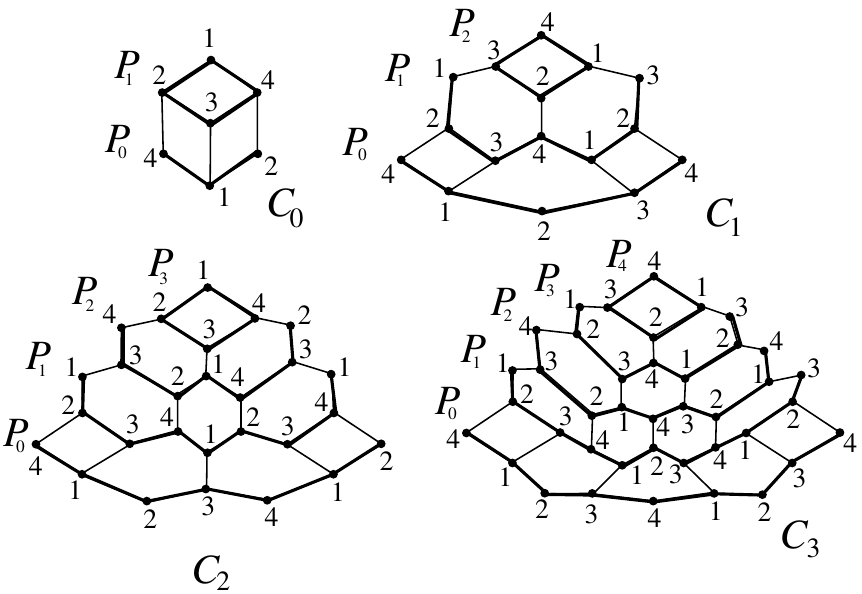}
\caption{The paths $P_i$ and the resulting distance-two colorings}
\label{messay2}
\end{figure}

We denote $P_i=v_i^0 v_i^1\cdots v_i^{\ell_i}$. The edges between $P_0$ and
$P_1$ are $v_0^0 v_1^1$, $v_0^j v_1^{j+1}$ for $1\leq j\leq \ell_0-1$, $j$
odd, and $v_0^{\ell_0} v_1^{\ell_1-1}$. We can choose the permutation of
colors for the square $v_0^0 v_0^1 v_1^2 v_1^1$ to be $4132$, forcing
for neighbors $v_1^0,v_1^3,v_0^2$ the colors $1,4,2$, and completing the
adjacent square, or hexagon with the assignment to $v_1^4,v_0^3$ of colors
$1,3$. This forced process extends similarly through the chain of hexagons
until the last square.

We have derived the beginning of $P_0$ as $4123$ and the beginning of $P_1$
as $12341$. After the forced extension, $P_0$ will be an initial segment
of ${(4123)}^*$ and $P_1$ will be an initial segment
of ${(1234)}^*$.

For $i\geq 1$, the edges between $P_i$ and
$P_{i+1}$ are $v_i^0 v_{i+1}^1$,
$v_i^j v_{i+1}^{j-1}$ for $3\leq j\leq \ell_i-3$, $j$
odd, and $v_i^{\ell_i} v_{i+1}^{\ell_{i+1}-1}$.

A similar process derives the beginning of $P_i$ for $i$ odd as $1234$
and the beginning of $P_i$ for $i$ even
as $4321$. After the forced extension, $P_i$ for $i$ odd
will be an initial segment
of ${(1234)}^*$ and $P_i$ for $i$ even will be an initial segment
of ${(4321)}^*$.

This gives a unique coloring for the triangular region after coloring one
square $S$, which is uniquely extended to the four triangular regions
surrounding $S$, and then uniquely extended to the four triangular regions
surrounding $S'$ opposite to $S$.
\end{proof}

\begin{theorem}\label{veta}
\label{gd}
The Goodey graphs $C_k, k\geq 0,$ are the only bipartite cubic planar graphs
having a distance-two four-coloring.
\end{theorem}

\begin{proof}
Consider a Goodey graph $G$ with a fixed distance-two four-coloring. Recall
that Goodey graphs have exactly six squares. Each of the squares is joined 
by four chains of hexagons to four squares. We consider the dual six-vertex 
graph $G'$ whose vertices are squares in $G$, with $ab$ an edge in $G'$ if
and only if there is a chain of hexagons joining squares $a$ and $b$. It can 
be readily verified that such a chain cannot cross itself or another chain in $G$.
Indeed, the colors in the fixed distance-two four-coloring are uniquely forced
along such chains and they don't match if the chains should cross. It follows
that the graph $G'$ is planar. A similar argument shows that a chain cannot
return to the same square, and two chains from the square $a$ cannot end
at the same square $b$. Thus $G'$ has no faces of size one or two, and by
Euler's formula it has $12$ edges and $8$ faces; therefore all faces of $G'$
must be triangles, and $G'$ is the octahedron.

Let $T$ be a triangular face in $G'$, let $s$ be a side of $T$ with the smallest 
number $d$ of hexagons in $G$. Then it can again be checked using the coloring
that the other two sides of $T$ will also have $d$ hexagons in $G$. Then $T$ 
corresponds to a triangular region $R$ as in Theorem~\ref{cd}, and the 
octahedron $G'$ yields $G=C_k$ for $k=d$.
\end{proof}

We can therefore conclude the following.

\begin{cor}
The distance-two four-coloring problem for Goodey graphs is solvable
in polynomial time.
\end{cor}

Recognizing whether an input Goodey graph is some $C_k$ can be
achieved in polynomial time; in the same time bound $G$ can actually
be distance-two four-colored.

\section{Distance-two four-coloring of type-two Barnette graphs is polynomial}

We now return to general type-two Barnette graphs, i.e., cubic plane graphs
with face sizes $3, 4, 5$, or $6$. As a first step, we analyze when a general
cubic plane graph admits a distance-two four-coloring which has three colors
on the vertices of every face of $G$.

\begin{theorem}\label{th}
A cubic plane graph $G$ has a distance-two four-coloring with three colors per 
face if and only if 

\begin{enumerate}
\item
all faces in $G$ have size which is a multiple of $3$,
\item
$G$ is bi-connected, and
\item
if two faces share more than one edge, the relative positions of the shared edges 
must be congruent modulo $3$ in the two faces.
\end{enumerate}
\end{theorem}

The last condition means the following: if faces $F_1, F_2$ meet in edges $e, e'$
and there are $n_1$ edges between $e$ and $e'$ in (some traversal of) $F_1$, 
and $n_2$ edges between $e$ and $e'$ in (some traversal of) $F_2$, then 
$n_1 \equiv n_2 (\mod 3)$.

\begin{proof}
Suppose $G$ has a distance-two four-coloring with three colors in each face.
The unique way to distance-two color a cycle with colors $1, 2, 3$ is by repeating
them in some order $(123)^*$ along one of the two traversals of the cycle. Therefore
the length is a multiple of $3$ so (1) holds. Moreover, there can be no bridge in $G$
as that would imply a face that self-intersects and is traversed in opposite directions
along any traversal of that face, disagreeing with the order $(123)^*$ in one of them;
thus (2) also holds. Finally, (3) holds because the common edges must have the same
colors in both faces.

Suppose the conditions hold, and consider the dual $G^D$ of $G$. (Note that each face
of $G^D$ is a triangle.) We find a distance-two coloring of $G$ as follows. Let $F$ be a
face in $G$; according to conditions (1-2), its vertices can be distance-two colored with
three colors. That takes care of the vertex $F$ in $G^D$. Using condition (3), we can
extend the coloring of $G$ to any face $F'$ adjacent to $F$ in $G^D$. Note that we
can use the fourth colour, $4$, on the two vertices adjacent in $F'$ to the two vertices
of a common edge. In this way, we can
propagate the distance-two coloring of $G$ along the adjacencies in $G^D$. If this
produces a distance-two coloring of all vertices of $G$, we are done. Thus it remains
to show there is no inconsistency in the propagation. If there is an inconsistency, it will
appear along a cycle $C$ in $G^D$. If there is only one face inside of $C$, then $C$ 
is a triangle corresponding to a vertex of $G$, and there is no inconsistency. Otherwise 
we can join some two vertices of $C$ by a path $P$ inside $C$, and the two sides of 
$P$ inside $C$ give two regions that are inside two cycles $C',C''$. The consistency of 
$C$ then follows from the consistency of each of $C',C''$ by induction on the number of 
faces inside the cycle.
\end{proof}

It turns out that conditions (1 - 3) are automatically satisfied for cubic plane graphs 
with faces of sizes $3$ or $6$.

\begin{cor}
Type-two Barnette graphs with faces of sizes $3$ or $6$ are distance-two four-colorable.
\end{cor}

\begin{proof}
Such a graph must be bi-connected, i.e., cannot have a bridge, since no triangle or
hexagon can self-intersect. Moreover, only two hexagons can have two common edges, 
and it is easy to check that they must indeed be in relative positions congruent modulo 
$3$ on the two faces. (Since all vertices must have degree three.) Thus the result follows 
from Theorem \ref{th}.
\end{proof}

\begin{theorem}
Let $G$ be type-two Barnette graph.
Then $G$ is distance-two four-colorable if and only if it is one of the graphs 
$C_k, k\geq 0$, or all faces of $G$ have sizes $3$ or $6$.
\end{theorem}

\begin{proof}
If there are faces of size both $3$ and $4$ (and possibly size $6$), then there must 
be (by Euler's formula) two triangles and three squares, and as in the proof of 
Theorem~\ref{gd}, the squares must be joined by chains of hexagons, which is 
not possible with just three squares.

If there is a face of size $5$, then there is no distance-two four-coloring since all five 
vertices of that face would need different colors.
\end{proof}

\section{Distance-two coloring of quartic graphs}

A {\em quartic graph} is a regular graph with all vertices of degree four. Thus any
distance-two coloring of a quartic graph requires at least five colors.
A {\em four-graph} is a plane quartic graph whose faces have sizes $3$ or $4$.
The argument to view these as analogues of type-two Barnette graphs is as
follows. For cubic plane Euler's formula limits the numbers of faces that are
triangles, squares, and pentagons, but does not limit the number of hexagon
faces. Similarly, for plane quartic graphs, Euler's formula implies that such a 
graph must have $8$ triangle faces, but places no limits on the number of 
square faces.

We say that two faces are {\em adjacent} if they share an edge.

\begin{lemma}
\label{four}
If a four-graph can be distance-two five-colored, then every square face
must be adjacent to a triangle face. Thus $G$ can have at most $24$ square 
faces.
\end{lemma}

\begin{proof}
We view the numbers $1, 2, 3, 4$ modulo $4$, and number $5$ is separate.
Let $u_1 u_2 u_3 u_4$ be a square face that has no adjacent triangle face.
(This is depicted in Figure \ref{fsquares} as the square in the middle.)
Color $u_i$ by $i$. Let the adjacent square faces be $u_i u_{i+1} w_{i+1} v_i$.
One of $v_i,w_i$ must be colored $5$ and the other one $i+2$. Then either
all $v_i$ or all $w_i$ are colored $5$, say all $w_i$ are colored $5$, and all
$v_i$ are colored $i+1$. Then $v_i u_i w_i$ cannot be a triangle face, or
$w_i,w_{i+1}$ would be both colored $5$ at distance two. Therefore
$t_i v_i u_i w_i$ must be a square face. (In the figure, this is indicated
by the corner vertices being marked by smaller circles; these must exist
to avoid a triangle face.) This means that the original square is surrounded
by eight square faces for $u_1u_2u_3u_4$, and $t_i$ must have color $i+3$, 
since $u_i,v_{i+3},v_i,w_i$ have colors $i,i+1,i+2,5$.

But then there cannot be a triangle face $x_iv_iw_{i+1}$, since $x_i$ is within
distance two of $u_i,u_{i+1},$ $v_i,t_i,w_{i+1}$ of colors $i,i+1,i+2,i+3,5$,
so each of the adjacent square faces $u_i u_{i+1} w_{i+1} v_i$ for
$u_1u_2u_3u_4$ has adjacent square faces as well. This process of moving to
adjacent square faces eventually reaches all faces as square faces, contrary
to the fact that there are $8$ triangle faces. 
\end{proof}

\begin{figure}[h!]
\begin{center}
\includegraphics[height=4cm]{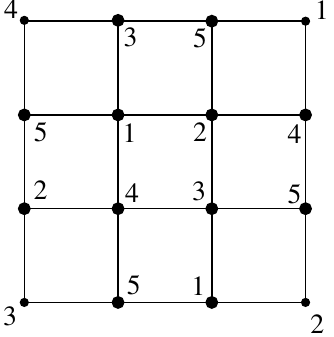}
\end{center}
\caption{One square without adjacent triangles implies all faces must be squares}
\label{fsquares}
\end{figure}

It follows that there are only finitely many distance-two five-colorable four-graphs.

\begin{cor}
The distance-two five-coloring problem for four-graphs is polynomial.
\end{cor}

In fact, we can fully describe all four-graphs that are distance-two five-colorable. 
Consider the four-graphs $G_0, G_1$ given in Figure \ref{fgzero}. The graph 
$G_0$ has 8 triangle faces and 4 square faces, the graph $G_1$ has 8 triangle 
faces and 24 square faces. Note that $G_0$ is obtained from the cube by inserting 
two vertices of degree four in two opposite square faces. Similarly, $G_1$ is obtained 
from the cube by replacing each vertex with a triangle and inserting into each face of 
the cube a suitably connected degree four vertex. (In both figures, these inserted 
vertices are indicated by smaller size circles.)

\begin{figure}[h!]
\includegraphics[height=4.5cm]{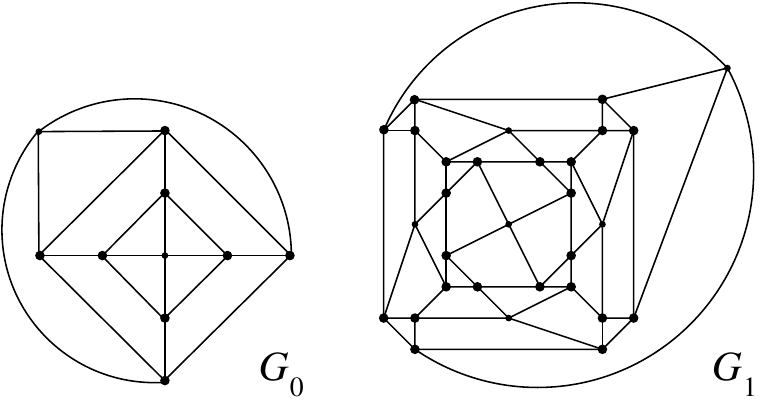}
\caption{The only four-graphs that admit a distance-two five-coloring}
\label{fgzero}
\end{figure}

\begin{theorem}\label{prsi}
The only four-graphs $G$ that can be distance-two five-colored are $G_0, G_1$.
These two graphs can be so colored uniquely up to permutation of colors.
\end{theorem}

\begin{proof}
We show that if $G$ can be so colored, then either $G$ is $G_0$ or every triangle 
in $G$ must be surrounded by six square faces, in which case $G$ is $G_1$.

Suppose $G$ has two adjacent triangles $u_5 u_1 u_2$ and $T=u_5 u_2 u_3$.
The vertices adjacent to $T$ must be given two colors
other than those of $u_5, u_2, u_3$. If $T$ has two adjacent
squares, then it has five adjacent vertices, which must be given the two 
colors in alternation, a contradiction.
Similarly if $T$ is adjacent to three triangles then the three vertices adjacent
would need three new colors, a contradiction.

We may thus assume a triangle $u_5 u_3 u_4$. If there is a square $u_1 u_5 u_4 t$,
this square plus the two adjacent triangles would need six colors, a contradiction,
so $u_5 u_4 u_1$ is a triangle, completing $u_5$ adjacent to the four-cycle
$u_1 u_2 u_3 u_4$. Color $u_i$ with color $i$. Then the additional vertex
$v_i$ adjacent to $u_i$ for $1\leq i\leq 4$ must be given color $i+2$
(modulo 4), so these $v_i$ form a 4-cycle, and any additional vertex adjacent
to a $v_i$ must get color 5, so there is a single additional $v_5$ with color 5.
This gives a uniquely colored $G_0$, up to permutation of colors.

In the remaining case, each triangle $T'=u_1 u_2 u_3$ has adjacent squares
$u_i u_{i+1} w_{i+1} v_i$, with addition modulo $3$. The vertices $v_i, w_{i+1}$
must be given the two colors different from those of $T'$, and in alternation
around $T'$, so there cannot be a triangle $u_i v_i w_i$ else $w_i,w_{i+1}$
with the same color would be at distance two. So there are squares
$u_i v_i t_i w_i$, and $T'$ is surrounded by six squares.

By Lemma~\ref{four}, we must have a triangle adjacent to the square
$u_i v_i t_i w_i$, either $v_i t_i x_i$ or $w_i t_i y_i$, but not both
since six colors would be needed. Let such a triangle be $T_i$, and we link $T'$
to the three $T_i$. These triangles viewed as vertices linked form a cubic graph
without triangles $G'$, since a triangle face would be three triangles joined in
$G$, which would need to have only three squares inside by Lemma~\ref{four}.
The graph $G'$ has 8 vertices for the 8 triangles, so this graph is the
cube $C$. Replacing each vertex corresponding to a triangle by the 
corresponding triangle gives a graph $D$.

Suppose the triangles adjacent to $T'$ are $v_i t_i x_i$ for $1\leq i\leq 3$. 
Then going around a face of $C$ we notice only one vertex inside this face
by Lemma~\ref{four}, giving
the construction of $G_1$. If we assign to the vertex inside this face the
color $5$, we notice that the surrounding triangles in $D$ must use three
colors at most 4, and each must omit a different color of 4. This implies
that all vertices in centers of square faces must be 5, and only opposite
triangles for $C$ use the same 3 out of 4 colors. This proves existence
and uniqueness up to permutation of colors of the distance-two 5-coloring
of $G_1$.

Suppose instead the adjacent triangles are $T_1=w_1 t_1 y_1$, $T_2=v_2 t_2 x_2$,
and $T_3=v_3 t_3 x_3$. If there is no triangle $v_1 w_2 x$, then the three squares
$Q_1,Q_2,Q_3$ between $T_1$ and $T_2$ are respectively adjacent to squares
$Q'_1,Q'_2,Q'_3$, and $Q'_2$ must be adjacent to a triangle by Lemma~\ref{four}. 
There must be triangles at both ends of the $Q'_i$ and these are adjacent to
$T_1$ and $T_2$, a contradiction.

Finally, suppose again the adjacent triangles are $T_1=w_1 t_1 y_1$, 
$T_2=v_2 t_2 x_2$, and $T_3=v_3 t_3 x_3$, but there is a triangle $v_1 w_2 x$.
This triangle faces $T'$, and $T_1$ faces $T_3$. Triangles facing each other
give two diagonals in the square faces of $C$, which implies two opposite
faces without such diagonals in $C$, while the four sets of two diagonals
form a matching of the $8$ vertices of $C$. If the center of a face without
diagonals gets assigned $5$, then the adjacent triangles will be assigned
a subset of $1\leq i\leq 4$. Then joining the sets of two diagonals assigns
a $5$ to a vertex of each remaining triangle, which is not possible to the
center of the remaining face without diagonals.
\end{proof} 

\begin{figure}[h!]
\begin{center}
\includegraphics[height=4cm]{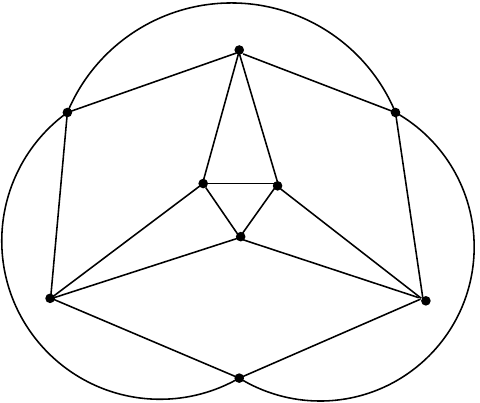}
\end{center}
\caption{A four-graph requiring nine colors in any distance-two coloring}
\label{weg}
\end{figure}

We close with a few remarks and open problems.

Wegner's conjecture \cite{weg} that any planar graph with maximum degree 
$d=3$ can be distance-two seven-colored has been proved in \cite{hjt,carsten}.
That bound is actually achieved by a type-two Barnette graph, namely the graph 
obtained from $K_4$ by subdividing three incident edges. Thus the bound of $7$ 
cannot be lowered even for type-two Barnette graphs.

Wegner's conjecture for $d=4$ claims that any planar graph with maximum 
degree four can be distance-two nine-colored. The four-graph in Figure \ref{weg} 
actually requires nine colors in any distance-two coloring. Thus if Wegner's 
conjecture for $d=4$ is true, the bound of $9$ cannot be lowered, even in the 
special case of four-graphs. It would be interesting to prove Wegner's conjecture 
for four-graphs, i.e., to prove that  {\em any four-graph can be distance-two 
nine-colored}.

Finally, we've conjectured that any bipartite cubic planar graph can be 
distance-two six-colored (a special case of a conjecture of Hartke, 
Jahanbekam and Thomas \cite{hjt}). The hexagonal prism (a cyclic
prism with $k=3$, which is a Goodey graph), actually requires six 
colors. Hence if our conjecture is true, the bound of $6$ cannot be
lowered even for Goodey graphs. It would be interesting to prove
our conjecture for Goodey graphs, i.e., to prove that {\em any Goodey 
graph can be distance-two six-colored}.

\newpage

\small
\bibliographystyle{abbrv}

\end{document}